\documentclass[12pt]{article}
\usepackage{amssymb}\usepackage{amsmath}
\usepackage{amsthm}
\usepackage{dsfont}
\usepackage[arrow, matrix, curve]{xy}
\usepackage{keyval}
\usepackage{xkeyval}
\usepackage{xcolor}
\usepackage{tikz}
\usepackage{times}

\newcommand{\beq}{\begin{equation}}
\newcommand{\eeq}{\end{equation}}
\newcommand{\bsp}{\begin{split}}
\newcommand{\esp}{\end{split}}
\newcommand{\bg}{\begin{gather}}
\newcommand{\eg}{\end{gather}}

\newcommand{\RR}{\mathbb R}

\newcommand{\A}{\mathcal{A}}
\newcommand{\D}{\mathcal{D}}
\newcommand{\T}{\mathcal{T}} 
\newcommand{\W}{\mathcal{W}}

\newcommand{\supp}{{\rm supp\>}}
\renewcommand{\labelenumi}{(\roman{enumi})}
\newtheorem{prop}{Proposition}
\newtheorem{thm}[prop]{Theorem}

\begin{document}
\title{The time slice axiom in perturbative quantum field theory on globally hyperbolic spacetimes}
\author{
Bruno Chilian{\footnote{e-mail: {\tt \small bruno.chilian@desy.de}
} } and Klaus Fredenhagen\footnote{e-mail: {\tt \small klaus.fredenhagen@desy.de}}  \\[2mm]
II. Institut f\"ur Theoretische Physik\\
Universit\"at Hamburg\\
D-22761 Hamburg, Germany\\
}

\maketitle
\begin{abstract}
The time slice axiom states that the observables which can be measured within an arbitrarily small time interval suffice to predict all other observables. While well known for free field theories where the validity of the time slice axiom is an immediate consequence of the field equation it was not known whether it also holds in generic interacting theories, the only exception being certain superrenormalizable models in 2 dimensions. In this paper we prove that the time slice axiom holds at least for scalar field theories within formal renormalized perturbation theory.

\end{abstract}

\section{Introduction}\setcounter{equation}{0}
It is an important feature of hyperbolic differential equations that they typically admit a well posed initial value problem. From the point of view of physics, this allows to predict the future from the present and has an enormous impact, both philosophically and in view of technical applications. It is therefore an important question whether this property remains valid in quantum field theory. In non-interacting theories the field equations are linear and the quantum fields obey the same equations as the classical fields which allows the same conclusions. Note, however, that a formal proof was given not earlier as in \cite{FullingNarcowichWald}. Note also that the proof of the well-posedness of the Cauchy problem for classical linear hyperbolic field equations on generic hyperbolic spacetimes was available up to recently only via mimeographed notes of Leray \cite{Leray} (see \cite{Baer} for a recent monograph).

In the case of quantum field theory the field equation has a somewhat unclear status due to the singular nature of quantum fields. Roughly speaking, the quantum fields are operator valued distributions and their nonlinear functions are notorically ill defined.
In renormalized perturbation theory quantum fields including composite fields can be defined as formal power series. The renormalized field equation, however, does not seem to be sufficient to answer the question of predictability.
Since the time zero fields are not always well defined even the formulation of the initial value problem seems to pose problems. The situation is better in $P(\varphi)_2$-models  \cite{GJ} and in the $\text{Yukawa}_2$-model \cite{Schrader} where the time slice axiom follows from the finite speed of propagation.

In axiomatic quantum field theory one therefore has weakened the requirement of a well posed initial value problem. Instead one requires that the algebra generated by all fields which can be measured within an arbitrarily small time slice is already the algebra of all observables. It is obvious that a precise formulation of the axiom needs a definition of an ``algebra  generated by a set of fields''. In algebraic field theory one may associate von Neumann algebras to spacetime regions with compact closure. The time slice axiom then corresponds to the axiom of primitive causality \cite{HK} which was first introduced by Haag and Schroer in \cite{HS}. In \cite{Garber}, it was shown that the diamond property is equivalent to the time slice axiom for the generalized free field.

In perturbative quantum field theory the difficulty of equipping the algebra with a suitable topology was a big obstacle for progress. But while studying the generic features of quantum field theories on curved spacetimes an algebra of functionals of classical fields was discovered \cite{DF01,HW} which carries a natural topology induced by the so-called microlocal spectrum condition \cite{Radzikowski,BFK}. This algebra, on the one hand, is an extension of the algebra of canonical commutation relations for the free field, hence the first step in  our proof is to show that the time slice axiom holds also on the larger algebra. On the other hand, the algebra contains the time ordered products of all Wick polynomials of the free field and therefore the coefficients of the formal power series defining the interacting fields. We show in the second step of our proof that the property of causal factorization of time ordered products which is the crucial starting point of causal perturbation theory \`a la Epstein-Glaser \cite{EG} implies that the time slice property remains valid for interacting field theories provided it holds true for free field theories.

Note that we prove the validity of the axiom within the framework of algebras of observables. The axiom then holds automatically in every Hilbert space representation of the theory. Since on a generic spacetime there is no distinguished representation, and since even in Minkowski space the choice of the vacuum representation of the interacting theory would require a control of the adiabatic limit, this is an enormous advantage compared to the formulation as irreducibility in the vacuum representation (i.e. triviality of the weak commutant) which was chosen in \cite{Streater-Wightman}.

Notations and conventions: Throughout this text, $M$ denotes a globally hyperbolic spacetime. We denote by $V_{\pm}$ the union $\bigcup_{x\in M}V_x^{\pm}\subset T^*M$ of all forward or backward lightcones $V_x^{\pm}$, respectively. By $J_{\pm}(N)$ we denote the causal future / past of any $N\subset M$, i.e. for any point $y\in J_ {\pm}(N)$, there exists a point $z\in N$ and a smooth, \emph{causal} (future / past directed) curve from $z$ to $y$. By $I_{\pm}$, we denote the chronological future resp. past (same definition, except that ``causal" is replaced by ``timelike"). We call a subset $P\subset M$ \emph{past compact} if $J_-(x)\cap P$ is contained in a compact set for all $x\in M$.
\section{The time slice axiom for the algebra of Wick polynomials}
\label{sect:wick}
In this section we prove the validity of the time slice axiom for the net $\mathcal{W}$ of algebras $\mathcal{W}(M)$ of Wick polynomials in the globally hyperbolic spacetimes $M$. This will set the stage for the treatment of the time slice axiom in perturbatively constructed interacting theories in section \ref{sect:interacting}, since $\mathcal{W}(M)$ also contains the time-ordered products which are used in the perturbation series. The technical preliminaries which are neccessary for this construction were given by Brunetti, Fredenhagen and K\"ohler in \cite{BFK}, and the construction of $\mathcal{W}(M)$ was carried out by D\"utsch and Fredenhagen in \cite{DF01} and Hollands and Wald in \cite{HW}.

Let $\A(M)$ be the algebra of the free field which is generated by smeared fields $\varphi(f)$, $f\in\D(M)$, which satisfy the relations
\begin{equation}
\begin{split}
&f\mapsto\varphi(f)\text{ is linear}\ , \\
&\varphi(f)^*=\varphi(\overline{f})\ , \\
&\varphi(K f) = 0\ , \\
&\left[ \varphi(f),\varphi(g) \right] = i\Delta(f,g)\ .
\end{split}
\end{equation}
Here $K=\Box -\kappa\mathbf{R} + m^2$ denotes the Klein-Gordon operator with the covariant wave operator $\Box=\nabla_{\mu}\nabla^{\mu}$ for the Levi-Civita connection $\nabla$ on $M$, the coupling $\kappa$ to the scalar curvature $\mathbf{R}$ and the mass $m$. In the following, $K_i f$, $i=1,\dots,n$ will denote the application of $K$ to an $f\in\D(M^n)$ with respect to the $i$-th argument. By $\Delta$ we denote the difference of the uniquely determined advanced and retarded fundamental solutions of $K$
.

Now, for the construction of the algebra of Wick polynomials $\W(M)$, let $\omega$ be a quasifree Hadamard state on $\A(M)$ and let $\omega_2$ be its two-point distribution. Normally ordered products are defined as the operator-valued distributions
\begin{equation}
:\varphi(x_1)\cdots\varphi(x_n):\stackrel{def}{=}\left. \frac{\delta^n}{i^n\delta f(x_1)\cdots\delta f(x_n)}
\exp\left[ \frac{1}{2}\omega_2(f\otimes f) + i\varphi(f) \right] \right|_{f=0}
\end{equation}
on the GNS Hilbert space corresponding to $\omega$. 
The algebra of Wick polynomials $\W(M)$ is now defined to consist of smeared normally ordered products,
\begin{equation}\label{eq:defSmearedWickProd}
\phi^{\otimes n}(f)=\int:\varphi(x_1)\cdots\varphi(x_n):f(x_1,\dots,x_n)dx_1\dots dx_n\ ,
\end{equation}
where the space of admissible test distributions $f$ is 
\begin{equation}\label{eq:condWF}
\T^n(M) =\left\{ f\in\D'(M^n) \text{ symm. }, \supp f\text{ compact},  WF(f)\cap\overline{V_-^n\cup V_+^n}=\emptyset \right\}\ .
\end{equation}

By Wick's theorem, the product of smeared normally ordered products is
\begin{equation}\label{prod}
\phi^{\otimes n}(f)\phi^{\otimes m}(g)=\sum_{k=0}^{\text{min}(n,m)}\phi^{\otimes(n+m-2k)}(f\otimes_k g) \ ,
\end{equation}
with the symmetrized, $k$-times contracted tensor product
\begin{equation}
\label{contrTensorProd}
\begin{split}
(f & \otimes_k g)(x_1,\dots, x_{n+m-2k})
\\& \stackrel{def}{=}\mathbf{S}\frac{n!m!}{(n-k!)(m-k)!k!}\int_{M^{2k}} dy_1\cdots dy_{2k}
\\& \omega_2(y_1,y_2)\cdots\omega_2(y_{2k-1},y_{2k})
f(x_1,\dots,x_{n-k},y_1,y_3,\dots,y_{2k-1}) \times
\\&g(x_{n-k+1},\dots,x_{n+m-2k},y_2,y_4,\dots,y_{2k}) \ ,
\end{split}
\end{equation}
where $\mathbf{S}$ denotes symmetrization in $x_1,\dots,x_{n+m-2k}$. The condition  on the wavefront set of the smearing distributions $f$ guarantees that the elements \eqref{eq:defSmearedWickProd} and the product are well defined. Moreover, the expressions \eqref{contrTensorProd} are again in $\T^{n+m-2k}(M)$ ~\cite{BFK}.

At this point, we would like to emphasize the observation made in \cite{DF01}, that the algebraic structure of this construction is independent of the particular Hadamard state $\omega$, and therefore, of the associated GNS-Hilbert space representation. In fact, \eqref{prod} may be taken as the definition of an associative product on the space $\T^{\bullet}(M)\stackrel{def}{=}\bigoplus_{n=0}^{\infty}\T^n(M)$, i.e.
\begin{equation}\label{star}
\left( f\star g\right)_n= \sum_{m+l-2k=n} f_m\otimes_k g_l\ .
\end{equation}
Now equation \eqref{eq:defSmearedWickProd} may be interpreted as the definition of a representation map
\begin{equation}
\begin{split}
\pi : \T^{\bullet}(M) & \rightarrow \W(M)\\
f & \mapsto \phi^{\otimes \bullet}(f)\ .
\end{split}
\end{equation}
We denote by $\T^{\bullet}_K(M)$ the ideal in $\T^{\bullet}(M)$ which is generated by elements $\mathbf{S}K_1 f$ for $f\in\T^{\bullet}(M)$. Due to the validity of the field equation, we see that $\T^{\bullet}_K(M)$ is just the kernel of $\pi$. Therefore, we have a faithful representation $\pi_{\omega}$ of the quotient algebra $\T^{\bullet}(M)/\T^{\bullet}_K(M)$ by elements in $\W(M)$:
\begin{equation}
\pi_{\omega}: f+\T^{\bullet}_K(M) \mapsto \phi^{\otimes \bullet}(f)\ .
\end{equation}

It has been found that it is sometimes more convenient to work in the abstract algebra $\T^{\bullet}(M)$ than in its realization $\W(M)$, especially when treating problems of renormalization. This is referred to as the off-shell formalism, since the validity of the field equation is not enforced in $\T^{\bullet}(M)$. However, since the validity of the time slice axiom obviously depends on the field equation, in this paper we work in the on-shell formalism, i.e. in the isomorphic algebras $\W(M)$ and $\T^{\bullet}(M)/\T^{\bullet}_K(M)$.

The algebra $\T^{\bullet}(M)$ carries a natural topology inherited from the H\"or\-man\-der topologies on the spaces of distributions $\T^n(M)$ with compact support and the given restrictions on the wave front set \cite{DuistermaatHoermander2}. In this topology the space of test functions $\D(M^n)$ is sequentially dense in $\T^n(M)$, the product (\ref{star}) is sequentially continuous and the ideal $\T^{\bullet}_K(M)$ is sequentially closed. As a consequence, convergence of sequences in $\W(M)$ is well defined and $\A(M)$ is sequentially dense in $\W(M)$.

The key ingredient that we need for the first part of our proof is the following
\begin{prop}\label{prop:01}
Let $f\in\T^n(M)$. Let $N$ be a neighborhood of a Cauchy surface $\Sigma$ in the past of $\supp f$.
Then there exists a $g\in\T^n(M)$ with the following properties:
\renewcommand{\theenumi}{\roman{enumi}}
\renewcommand{\labelenumi}{(\theenumi)}
\begin{enumerate}
\item $g=f+i$ where  $i\in\T^{\bullet}_K(M)\cap\T^n(M)$ \label{Wick:item:01}
\item $\supp g\subset N$ \ . \label{Wick:item:02}
\end{enumerate}
\end{prop}

\begin{proof}
There exists a compact $C\subset M$ such that $\supp f\subset C^n$. Let $\Sigma_1$ be a Cauchy surface in the past of $C$ and $\Sigma_0$ be another Cauchy surface in the past of $\Sigma_1$. Let $\chi\in\mathcal{C}^{\infty}(M)$ be such that
\begin{equation}
\begin{split}
\chi(x)=\left\{
\begin{array}{cl}
1 & \text{for } x\in J_+(\Sigma_1)\ , \\
0 & \text{for } x\in J_-(\Sigma_0)\ .
\end{array}
\right.
\end{split}
\end{equation}
Using the advanced fundamental solution $\Delta^{av}$ of the Klein-Gordon equation, we define an operator $\alpha_i:\T^n(M)\rightarrow\T^n(M)$ by
\begin{equation}
\begin{split}
(\alpha_i f)(x_1,\dots,x_n)=&f(x_1,\dots,x_n)-K_i\,\chi(x_i)\\
&\int dy \,\Delta^{av}(x_i,y)f(x_1,\dots,x_{i-1},y,x_{i+1},\dots,x_n) \ .
\end{split}
\end{equation}
We now claim that for $g = \alpha_1\cdots\alpha_n f$, we have $g\in\T^n(M)$, and that $g$ has the properties (\ref{Wick:item:01}) and (\ref{Wick:item:02}).

To prove that $g\in\T^n(M)$, we have to check the symmetry, support and wavefront set of $g$. The symmetry is obvious from the definition of $g$. That $g$ has compact support, follows from Proposition \ref{prop:02}, together with the compactness of $\supp f$, the support properties of $\Delta^{av}$, and the fact that $\supp \chi$ is past compact.

For the wave front property, it is sufficient to show that $WF\left(\Delta^{av}_i f\right)\cap\overline{V_-^n\cup V_+^n}=\emptyset$ for
\begin{equation}
\left(\Delta^{av}_i f\right)(x_1,\dots,x_n)=\int dy \Delta^{av}(x_i,y) f(x_1,\dots,x_{i-1},y,x_{i+1},\dots,x_n) \ ,
\end{equation}
since multiplication by a smooth function and application of a differential operator do not enlarge the wavefront set of a distribution.
$K_i$ is a properly supported differential operator on $M^n$ with real principal part
\begin{equation}
k_i(X,\Xi)=g_{x_i}(\xi_i,\xi_i)
\end{equation}
for $(X,\Xi)\in T^{\ast}(M^n)$,  $X=(x_1,\dots,x_n)$, $\Xi=(\xi_1,\dots,\xi_n)$ and the metric $g$ on $M$.
By definition, there holds $K_i\Delta^{av}_i f=f$, so we may apply H\"ormander's theorem on the propagation of singularities  ~\cite[Theorem 6.1.1]{DuistermaatHoermander2}. Suppose that $(X_0,\Xi_0)\in WF(\Delta^{av}_i f)\cap\overline{V_-^n\cup V_+^n}$. By H\"ormander's theorem, it follows that $(X_0,\Xi_0)\in \text{Char } K_i$, i.e. $k_i(X_0,\Xi_0)=0$. Let $\gamma$ be the uniquely determined curve in $T^{\ast}(M^n)$ such that $\gamma(0)=(X_0,\Xi_0)$ and $\gamma '(t)=H_k(\gamma(t))$ with the Hamiltonian vector field $H_k$ associated to the Hamiltonian function $k_i$.
$\gamma$ has the form
\begin{equation}
\gamma(t)=\left(\left(q_1(t),\dots,q_n(t) \right), \left(p_1(t),\dots,p_n(t)\right)\right)
\end{equation}
with $\left(q_j(t),p_j(t) \right)\in T^{\ast}(M) \quad \forall j\in\{1,\dots,n \}, t\in\RR$.
Following ~\cite[Proposition 4.2]{Radzikowski} $\left(q_i,p_i \right)$ is a null geodesic strip in $T^{\ast}(M)$.
Since $k(X,\Xi)$ depends only on $x_i$ and $\xi_i$, only the $i$-th component of $H_k(X,\Xi)$ is nonzero, i.e. $\left( q_j(t),p_j(t)\right)=\left( q_j(0),p_j(0)\right)$ $\forall t\in \RR$ for $i\neq j$. So we have $\gamma(t)\in\overline{V_-^n\cup V_+^n}$ $\forall t\in\RR$. Due to the assumption on $WF(f)$ we therefore have $\gamma(\RR)\cap WF(f)=\emptyset$. So from H\"ormander's theorem it follows that
\begin{equation}
\label{Wick:eq:19}
\gamma(\RR)\subset WF(\Delta^{av}_i f) \ .
\end{equation}
There exists a Cauchy surface $\Sigma$ with  $\supp \Delta^{av}_i f\subset \left(J_-(\Sigma)\right)^n$.
Since $q_i$ is a causal curve in $M$, $J_+(\Sigma)\cap q_i(\RR)\neq\emptyset$. So
\begin{equation}
\left(q_1(\RR),\dots,q_n(\RR)\right) \nsubseteq \supp \Delta^{av}_i f \ .
\end{equation}
But this contradicts \eqref{Wick:eq:19}, so there must hold $WF(\Delta^{av}_i f)\cap\overline{V_-^n\cup V_+^n}=\emptyset$.

Property (\ref{Wick:item:01}) is now obvious from the definition of the $\alpha_i$.
To prove (\ref{Wick:item:02}), we define a family of open subsets $\left( J_i^s \right)_{(i,s)\in\{1,\dots,n \}\times\{+,-\}}$ by
\begin{equation}
\label{eq:03}
\begin{split}
J_i^+=& \big\{ (x_1,\dots,x_n)\in M^n \big| x_i \in J_+(\Sigma_{t_1})\setminus \Sigma_{t_1}  \big\},\\
J_i^-=& \big\{ (x_1,\dots,x_n)\in M^n \big| x_i \in J_-(\Sigma_{t_0})\setminus\Sigma_{t_0} \big\} \ .
\end{split}
\end{equation}
This is an open cover of $M^n \setminus \left( J_-(\Sigma_{t_1})\cap J_+(\Sigma_{t_0})\right)^n$ so there exists a partition of unity  $\left( \varepsilon_i^s \right)_{(i,s)\in\{1,\dots,n \}\times\{+,-\}}$, such that
\begin{equation}
\label{eq:04}
\begin{split}
 &\varepsilon_i^s\in \mathcal{C}^{\infty}\left(M^n\setminus \left( J_-(\Sigma_{t_1})\cap J_+(\Sigma_{t_0})\right)^n\right),\\
 &\sum_{i,s} \varepsilon_i^s(x)=1 \quad\forall x\in M^n \setminus \left( J_-(\Sigma_{t_1})\cap J_+(\Sigma_{t_0})\right)^n,\\
 &\supp \varepsilon_i^s \subset J_i^s \ .
\end{split}
\end{equation}
For $\Phi\in\mathcal{C}^{\infty}(M^n)$ with $\supp \Phi\cap \left( J_-(\Sigma_{t_1})\cap J_+(\Sigma_{t_0})\right)^n=\emptyset$, we have
\begin{equation}
\label{eq:5}
 (\alpha_1\dots\alpha_n f_n)(\Phi)
 =\sum_i (\alpha_1\dots\alpha_n f_n)(\varepsilon_i^+\Phi) + \sum_i (\alpha_1\dots\alpha_n f_n)(\varepsilon_i^-\Phi) \ .
\end{equation}
Let $i\in\{1,\dots,n \}$, then
\begin{equation}
\begin{split}
&(\alpha_1\dots\alpha_n f_n)(\varepsilon_i^+\Phi)\\
=&(\alpha_1\dots\widehat{\alpha_i}\dots\alpha_n f_n)(\varepsilon_i^+\Phi) -
\left(K_i\chi_i(\Delta^{av}_i\alpha_1\dots\widehat{\alpha_i}\dots\alpha_n f_n)\right)(\varepsilon_i^+\Phi)\\
=&(\alpha_1\dots\widehat{\alpha_i}\dots\alpha_n f_n)(\varepsilon_i^+\Phi) -
\left(K_i(\Delta^{av}_i\alpha_1\dots\widehat{\alpha_i}\dots\alpha_n f_n)\right)(\varepsilon_i^+\Phi)\\
=&0 \ ,
\end{split}
\end{equation}
where we used the notation $\chi_i(x_1,\dots,x_n)=\chi(x_i)$ and the fact that $\chi_i\varepsilon_i^+=\varepsilon_i^+$.
Now we consider
\begin{equation}
\begin{split}
&(\alpha_1\dots\alpha_n f_n)(\varepsilon_i^-\Phi)\\
=&(\alpha_1\dots\widehat{\alpha_i}\dots\alpha_n f_n)(\varepsilon_i^-\Phi) -
\left(K_i\Delta^{av}_i\alpha_1\dots\widehat{\alpha_i}\dots\alpha_n f_n)\right)(\chi_i\varepsilon_i^-\Phi) \ .
\end{split}
\end{equation}
The terms on the right hand side vanish individually. The first one vanishes, because
$\supp \varepsilon_i^-\Phi\subset M^{i-1}\times J_-(\Sigma_{t_0})\times M^{n-i}$, but
$\supp \alpha_1\dots\widehat{\alpha_i}\dots\alpha_n f_n \subset M^{i-1}\times X \times M^{n-i}$
and $X\cap J_-(\Sigma_{t_0}) = \emptyset$. The second term vanishes, since $\chi_i\varepsilon_i^-=0$. So all terms in \eqref{eq:5} vanish.
\end{proof}
As a consequence, we get the validity of the time slice axiom for the algebra $\mathcal{W}(M)$ of Wick polynomials:
\begin{thm}
Let $M$ be a globally hyperbolic spacetime and let $N$ be a neighborhood of a Cauchy surface $\Sigma$ of $M$. Then, 
\begin{equation}
\mathcal{W}(N)=\mathcal{W}(M) \ .
\end{equation} 
\end{thm}
\begin{proof}
Let $\mathcal{W}(M)\ni F(\varphi)=\sum_{n=0}^{k} \int:\varphi(x_1)\cdots\varphi(x_n):f_n(x_1,\dots,x_n)dx_1\dots dx_n$. Using Proposition \ref{prop:01} for each $f_n$, we can find $g_n\in\T^n(N)$, $n\in\{1,\dots,n\}$ with $\supp g_n \subset N$. So the $g_n$ define an element $G(\varphi)=\sum_{n=0}^{k} \int:\varphi(x_1)\cdots\varphi(x_n):g_n(x_1,\dots,x_n)dx_1\dots dx_n$ of $\mathcal{W}(N)$. But because of property (\ref{Wick:item:01}) in Proposition \ref{prop:01} and the validity of the wave equation in $\W(M)$, we have $F(\varphi)=G(\varphi)$.
\end{proof}

\section{The time slice axiom for the interacting case}
\label{sect:interacting}
In causal perturbation theory, the interacting fields are constructed in terms of time ordered products of polynomials of fields. Most easily one describes the time ordered products in terms of their generating functions, the so-called local S-matrices. They are formal power series with coefficients in the algebra of Wick polynomials, and they are functionals $S(g)$ of test functions $g\in\mathcal{D}(M,V)$ where $V$ denotes the finite dimensional space of possible interaction terms, including the free field itself. The crucial property of the S-matrix is the causal factorization property
\begin{equation}\label{causal}
S(f+g+h)=S(f+g)S(g)^{-1}S(g+h) 
\end{equation}
if the supports of $f$ and $h$ are causally separated in the sense that there is a Cauchy surface such that $\supp f$ is in the future and $\supp h$ is in the past of the surface.
Together with the normalization condition $S(0)=1$ this is the only property of perturbation theory which we need in our proof. The proof is therefore also valid beyond perturbation theory once a solution to the causal factorization property has been found. 

The algebra $\mathfrak{A}_0(\mathcal{O})$ of the free field associated to the region $\mathcal{O}$ is the sequentially closed unital $*$-algebra generated by the elements $S(g)$ where $\supp g\subset\mathcal{O}$.
We assume that this net of algebras satisfies the time slice axiom. (In perturbation theory this follows from the fact that this net of algebras is identical to net of algebras discussed in the previous section.) 

Interacting fields for localized interactions parametrized by a spacetime dependent coupling constant $g\in\mathcal{D}(M,V)$ can be defined in terms of relative S-matrices
$S_g(f):=S(g)^{-1}S(g+f)$. It is a nice feature of the causal factorization property that the relative S-matrices $S_g$ satisfy also the functional equation (\ref{causal}). Therefore another interaction may be added. In particular the original interaction may be compensated so that the free field can be expressed in terms of the interacting fields. This is the starting idea for the extension of the time slice property to the interacting case.

But there is an obstacle, namely we want to describe interactions which do not vanish outside of a compact region. They can no longer be described by test functions $g$
with compact support. Instead we have to admit smooth functions $g\in\mathcal{E}(M,V)$.
Fortunately, due to the functional equation (\ref{causal}), it turns out that the structure
of the algebras $\mathfrak{A}_g(\mathcal{O})$, which are taken to be the sequentially closed unital $*$-algebras generated by the relative S-matrices $S_g(f)$ with $\supp f\subset\mathcal{O}$, is independent of the behaviour of $g$ outside of $\mathcal{O}$. This allows to perform the limit to interactions with noncompact support in a purely algebraic way (``algebraic adiabatic limit'' \cite{BF}).

Actually, again due to (\ref{causal}) the relative S-matrices $S_g$ can be defined if the support   of $g$ is past compact. In this case, we may use Proposition \ref{prop:02} to conclude that the past of every compact set $K\subset M$ intersects $\supp g$ only within a relatively compact region. Consequently, we can find a test function $b\in\mathcal{D}(M,V)$ which coincides with $g$ in $J_-(K)$.
Then for every $f$ with support in the interior of $K$ we set
\begin{equation}
S_g(f)=S_b(f) \ .
\end{equation}
The right hand side does not depend on the choice of $b$, for let $\tilde{b}$ also satisfy the condition on $b$, then for $c=\tilde{b} - b$, $\supp c$ does not intersect the past of $K$. So we have
\begin{equation}\label{eq:independence}
S_{\tilde{b}}(f)=S(b+c)^{-1}S(b+c)S(b)^{-1}S(b+f)=S_b(f) \ .
\end{equation}

We now want to prove the time slice axiom. Let $\mathcal{O}$ be a relatively compact region and let $\Sigma$ be a Cauchy surface. We choose a second Cauchy surface $\Sigma_1$ such that $\Sigma$ and $\mathcal{O}$ lie in the interior of the future of $\Sigma_1$. Now let $g$ be a smooth function with support in the interior of the future of $\Sigma_1$. We want to prove that
$\mathfrak{A}_g(\mathcal{O})\subset\mathfrak{A}_g(N)$ for every open neighbourhood $N$ of $\Sigma$ . We already know by construction that $\mathfrak{A}_g(\mathcal{O})\subset\mathfrak{A}_0(M)$ and that by the time slice property for the free field $\mathfrak{A}_0(M)=\mathfrak{A}_0(N)$ for every open neighbourhood of $\Sigma$. Let $N'\subset N$ be an open neighbourhood of $\Sigma$ whose past boundary is again a Cauchy surface (see figure \ref{figure}).   We choose a function $g'$ with $\supp g'\subset N$ which coincides with $g$ on $N'$. We then construct the relative S-matrices
$$S_{g,g'}(f)=S_g(-b')^{-1}S_g(-b'+f)$$ 
with a test function $b'$ with support in $N$ and which coincides with $g'$ on $J_-(K)$, where $\supp f$ is contained in the interior of $K$ and $K\subset N'$.   Again, the right hand side is independent of the choice of $b'$. Inserting the definition of $S_g$ we obtain
$$S_{g,g'}(f)=S(b-b')^{-1}S(b-b'+f) \ $$
where $b$ coincides with $g$ on $J_-(L)$ and where $L$ is a compact region containing $\supp b'$ in its interior. Now we may split $b-b'=b_+ + b_-$ such that $\supp b_+$ does not intersect the past and $\supp b_-$ not the future of $\supp f$, hence the second factor factorizes, and we obtain
$$S_{g,g'}(f)=S(b_-)^{-1}S(f)S(b_-) \ .$$
\begin{figure}\centering
\begin{tikzpicture}[scale=1,>=stealth]

\fill [black!20!white] (-6,1) -- (6,1) -- (6,-1) -- (-6,-1) -- cycle; 
\fill [black!30!white] (-6,0.6) -- (6,0.6) -- (6,-0.6) -- (-6,-0.6) -- cycle; 

\fill [black!30!white] (-6,-2) -- (6,-2) -- (6,-2.5) -- (-6,-2.5) -- cycle; 
\draw (0,-2.25)node{$S$};

\draw [gray] (-6,1) -- (6,1); 
\draw [gray] (-6,0.6) -- (6,0.6); 
\draw (-6,0) -- (6,0); \draw (-2,0) node[fill=black!30!white]{$\Sigma$};
\draw [gray] (-6,-0.6) -- (6,-0.6); 
\draw [gray] (-6,-1) -- (6,-1); 
\draw [gray] (-6,-2) -- (6,-2); \draw (6.3,-2)node{$\Sigma_1$};
\draw [gray] (-6,-2.5) -- (6,-2.5); \draw (6.3,-2.5)node{$\Sigma_2$};

\draw (-.75,0) .. controls (-0.5,0.45) and (0.5,0.45) .. (.75,0); 
\draw (-.75,0) .. controls (-0.5,-0.45) and (0.5,-0.45) .. (.75,0); 
\draw (0,0) node[fill=black!30!white]{$K$};

\draw (-0.75,0) -- (-2.193,-2.5);
\draw (-2.193,-2.5) -- (-4.502,1.5);

\draw (0.75,0) -- (2.193,-2.5);
\draw (2.193,-2.5) -- (4.502,1.5);

\draw [semithick] (-4,0) .. controls (-4,.8) and (-3.5,0.8) .. (0,0.8) ; 
\draw [semithick] (-4,0) .. controls (-4,-1) and (-4,-1) .. (-3.5,-1) ; 
\draw (1.5,0.6) node {$\supp b'$};
\draw [semithick] (4,0) .. controls (4,.8) and (3.5,0.8) .. (0,0.8) ; 
\draw [semithick] (4,0) .. controls (4,-1) and (4,-1) .. (3.5,-1) ; 
\draw [semithick] (3.5,-1) -- (-3.5,-1); 

\draw [semithick] (-5.7,-1) .. controls (-5.7,1.5) and (-5,1.5) .. (0,1.5); 
\draw [semithick] (-5.7,-1) .. controls (-5.7,-1.5) and (-5.5,-2) .. (-5.3,-2); 
\draw [semithick] (5.7,-1) .. controls (5.7,1.5) and (5,1.5) .. (0,1.5); 
\draw [semithick] (5.7,-1) .. controls (5.7,-1.5) and (5.5,-2) .. (5.3,-2); 
\draw (1.5,1.3) node {$\supp b$};
\draw [semithick] (5.3,-2) -- (-5.3,-2); 

\draw (-6.3,0) .. controls (-6.1,.1) and (-6.4,.6) .. (-6.1,.6); 
\draw (-6.3,0) .. controls (-6.1,-.1) and (-6.4,-.6) .. (-6.1,-.6); 
\draw (-6.6,0) node {$N'$};

\draw (6.3,0) .. controls (6.1,.1) and (6.4,1) .. (6.1,1); 
\draw (6.3,0) .. controls (6.1,-.1) and (6.4,-1) .. (6.1,-1); 
\draw (6.6,0) node {$N$};
\end{tikzpicture}
\caption{Sketch of the geometrical configuration}
\label{figure}
\end{figure}
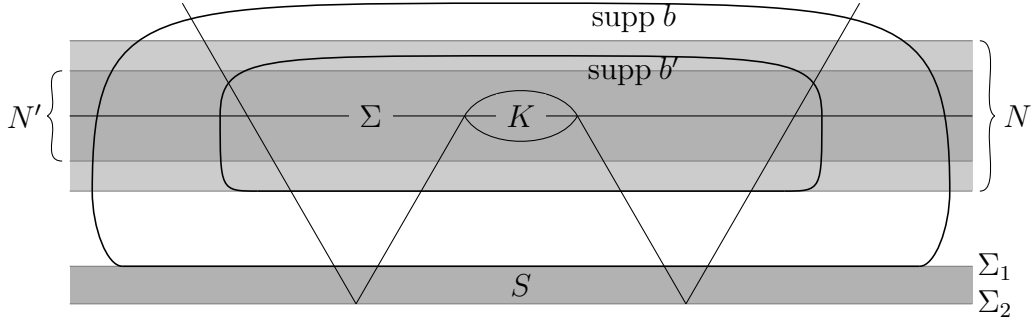

We draw two conclusions: First we see that for every $A\in\mathfrak{A}_0(\mathcal{O}_1)$ with $\mathcal{O}_1$ relatively compact and with closure $K\subset N'$ we have
\begin{equation}\label{a}
S(b_-)^{-1}A S(b_-)\in\mathfrak{A}_g(N)
\end{equation}
provided $b$ and $b'$ satisfy the conditions above. Moreover, due to the validity of the time slice axiom for the free theory, we have
\begin{equation}\label{b}
S(b_-)^{-1}A S(b_-)\in\mathfrak{A}_0(\tilde{\mathcal{O}})
\end{equation}
with  $\tilde{\mathcal{O}}=N''\cap J_+(J_-(K)\cap N)$ where $N''$ is an open neighbourhood of $\Sigma$ whose closure is contained in $N'$. 
We see that the map
\begin{equation}
S(f)\mapsto S_{g,g'}(f)\ ,\ \supp f\subset N'
\end{equation}
extends to an endomorphism $\alpha$ of $\mathfrak{A}_0(M)$ whose image is contained in $\mathfrak{A}_g(N)$. Moreover, for each relatively compact region $\mathcal{O}_2$ there is an invertible element $U\in \mathfrak{A}_0(M)$ such that
$\alpha(A)=UAU^{-1}$ for all $A\in\mathfrak{A}_0(\mathcal{O}_2)$.
In particular, $\alpha$ is injective. 
It remains to show that $\alpha$ is surjective. Since $g-g'$ vanishes within $N'$, we may decompose it into $g-g'=g_++g_-$ where $g_+$ has support in the future and $g_-$ in the past of
$N'$. For $f$ with $\supp f\subset N'$ we have
$\alpha(S(f))=S_{g_-}(f)$. It amounts to the statement that an interaction in the past causes an endomorphism of the algebra of observables which can be approximated by inner automorphisms $\alpha_b$ implemented by invertible elements $S(b_-)$ where $b_-$ coincides with $g_-$ in the past of a sufficiently large compact subregion of $\Sigma$. We may now choose $b_-$ such that it coincides with $g_-$ also in the future of $J_-(K)\cap S$, where $S=J_-(\Sigma_1)\cap J_+(\Sigma_2)$, and $\Sigma_2$ is a Cauchy surface in the past of $\Sigma_1$. By the time slice axiom for the algebra of the free field, we have $\mathfrak{A}_0(K)\subset\mathfrak{A}_0(J_-(K)\cap S)$. But $\mathfrak{A}_0(J_-(K)\cap S)$ is generated by elements $S(h)$ where $h\in\mathcal{D}(M,V)$ with $\supp h\subset J_-(K)\cap S$. For each such $h$, $\supp h$ is in the past of $\Sigma_1$ and $\supp b_-$ in the future, so
\begin{equation*}
\alpha_{b_-}^{-1}(S(h))  = S(b_-)S(h)S(b_-)^{-1} = S(b_-+h)S(b_-)^{-1}
\end{equation*}
does not depend on the choice of $b_-$ by an argument analogous to the one used in \eqref{eq:independence}.
So the inverse of $\alpha$ exists on $\mathfrak{A}_0(J_-(K)\cap S)\supset\mathfrak{A}_0(K)$ for all compact regions $K\subset N'$ and hence everywhere.

\section{Appendix}
The following simple proposition is used in two separate arguments of our proof, in sections \ref{sect:wick} and \ref{sect:interacting}, respectively.
\begin{prop}
\label{prop:02}
Let M be a globally hyperbolic spacetime, $P\subset M$ past compact and $K\subset M$ compact. Then $J_-(K)\cap P$ is contained in a compact set.
\end{prop}
\begin{proof}
Choose a Cauchy surface $\Sigma$ in the future of $K$. The family $\left(I_-(y)\right)_{y\in\Sigma}$ is an open cover of $K$, since $I_{\pm}(z)$ is open for any $z\in M$. Since $K$ is compact, there exists a finite subset $Y=\{y_1,\dots,y_n\}\subset\Sigma$ such that $\left(I_-(y)\right)_{y\in Y}$ is an open cover of $K$. 
Let $x$ be any point in $I_-(K)$. 
Then, $x\in I_-(y_i)$ for some $i\in \{1,\dots,n\}$. (This is true because there exists
a timelike, future directed curve $\gamma$, from $x$ to some $k\in K$. Since $k\in I_-(y_i)$ for some $i$, $\gamma$ can be extended timelike to $y_i$.)
So we have that
\begin{equation}
I_-(K)\subset \bigcup_{y\in Y}I_-(y) \subset \bigcup_{y\in Y}J_-(y)\ .
\end{equation}
Since $J_{\pm}(S)\subset \overline{I_{\pm}(S)}$ for any $S\subset M$, and since $J_-(z)$ is closed for any $z\in M$, 
\begin{equation}
J_-(K)\subset \overline{I_-(K)}\subset\overline{\bigcup_{y\in Y}J_-(y)} 
= \bigcup_{y\in Y}\overline{J_-(y)}
= \bigcup_{y\in Y}J_-(y)\ .
\end{equation}
Intersecting with $P$, we get
\begin{equation}
J_-(K)\cap P \ \subset\  \left(\bigcup_{y\in Y}J_-(y)\right)\cap P
\ \subset\  \bigcup_{y\in Y}\left(J_-(y)\cap P\right)\ , 
\end{equation}
where the right hand side is by assumption contained in a finite union of compact sets.
\end{proof}

\end{document}